\documentclass[a4paper,11pt]{article}

\usepackage[left=2.5cm,right=2.5cm,top=3cm,bottom=3cm,pdftex]{geometry}

\usepackage{amssymb}
\usepackage{amsmath}
\usepackage{amsthm}

\usepackage{dsfont}
\usepackage{url}
\usepackage{natbib}

\usepackage[utf8]{inputenc}
\usepackage[T1]{fontenc}

\usepackage[pdftex]{graphicx}
\DeclareGraphicsExtensions{.png, .pdf, .jpg} 

\usepackage[pdftex, colorlinks, linkcolor=blue, urlcolor=blue, citecolor=blue, breaklinks=true]{hyperref}

\newtheorem{thm}{Theorem}
\newtheorem{lem}{Lemma}
\newtheorem{fol}{Corollary}
\newtheorem{bem}{Remark}

\begin{document}

\begin{center}
\Large
\textbf{Optimizing the allocation of trials to sub-regions in multi-environment crop variety testing}
\end{center}

\begin{center}
\textbf{Maryna Prus and Hans-Peter Piepho}
\end{center}

\begin{abstract}
New crop varieties are extensively tested in multi-environment trials in order to obtain a solid empirical basis for recommendations to farmers. When the target population of environments is large and heterogeneous, a division into sub-regions is often advantageous. When designing such trials, the question arises how to allocate trials to the different sub-regions. We consider a solution to this problem assuming a linear mixed model. We propose an analytical approach for computation of optimal designs for best linear unbiased prediction of genotype effects and pairwise linear contrasts and illustrate the obtained results by a real data example from Indian nation-wide maize variety trials. It is shown that, except in simple cases such as a compound symmetry model, the optimal allocation depends on the variance-covariance structure for genotypic effects nested within sub-regions. 
\end{abstract}


\section{Introduction}

New crop varieties are usually evaluated for their performance in a target population of environments  (TPE), where environments correspond to locations in specific years. This evaluation requires conducting randomized field trials at several environments sampled from the TPE. Such trials are called multi-environment trials (MET). Analysis of MET is routinely performed using linear mixed models comprising effects for genotypes, environments and their interaction (\cite{isik}).

If the TPE is large and can be suitably stratified along geographical borders or agro-ecological zonations, it may be advantageous to subdivide the TPE into sub-regions. If the same set of genotypes is tested at a number of locations in each of the sub-regions, a linear mixed model may be fitted with random genotype-within-subregion effects that allows estimating a genotype's average performance in each sub-region using best linear unbiased prediction (BLUP) (\cite{atl}; \cite{pie1}). If a covariance is assumed between a genotype's performance in different sub-regions, the model allows borrowing strength across sub-regions, meaning that estimates of mean performance in a sub-region becomes more accurate than when based on data from the sub-region alone (\cite{kle}).

Whereas analysis of sub-divided TPE data has received some attention in the recent past, to the best of our knowledge the design of MET when a sub-division is envisaged has not been considered. The design of such trials has gained interest recently in endeavours to integrate trial networks across country borders (e.g., Horizon 2020 project INVITE = INnovations in plant VarIety Testing in Europe).

The design of MET for a sub-divided TPE involves two decisions: (1) The total number of environments at which to conduct the trials and (2) the allocation of this total number of environments to the different sub-regions. This paper is devoted to the second decision.

\section{Model Specification and Prediction}

In this work we use a linear mixed model (LMM) in which observation $l$ of genotype $k$ in location $j$ within the $i$-th sub-region is given by   
\begin{equation}\label{mod}
Y_{ijkl}=\mu_i+\alpha_{ik}+\lambda_{ij}+\gamma_{ijk}+b_{ijl}+\varepsilon_{ijkl} 
\end{equation}
for $l=1,2$, $k=1, \dots, K$, $i=1 \dots P$  and $j=t_{i-1}+1, \dots, t_i$, 
where $t_i=\sum_{s=1}^iJ_s$, $t_0=0$, $K$ denotes the number of genotypes, $P$ is the number of sub-regions, $J_i$ is the number of locations within the $i$-th sub-region and $J=\sum_{i=1}^PJ_i$ is the total number of locations. 
Moreover, $\mu_i$ denotes the mean (fixed) effect of the $i$-th sub-region, $\alpha_{ik}$ is the interaction effect of genotype $k$ in sub-region $i$, $\lambda_{ij}$ is the effect of the $j$-th location within the $i$-th sub-region, $\gamma_{ijk}$ denotes the effect of the $k$-th genotype in the $j$-th location within the $i$-th sub-region, $b_{ijl}$ is the effect of the $l$-th replication in location $j$ in sub-region $i$ and $\varepsilon_{ijkl}$ denotes the observational error. All random effects and observational errors are assumed to have zero mean. The variances are given by $\mathrm{var}(\varepsilon_{ijkl})=\sigma^2$, $\mathrm{var}(\lambda_{ij})=\sigma^2_{\lambda}=\sigma^2v_1$, $\mathrm{var}(\gamma_{ijk})=\sigma^2_{\gamma}=\sigma^2v_2$ and $\mathrm{var}(b_{ijl})=\sigma^2_{b}=\sigma^2v_3$ and the covariance matrix of the genotype effects $\mbox{\boldmath{$\alpha$}}_k=(\alpha_{1k}, \dots, \alpha_{Pk})^\top$ is $\mathrm{Cov}(\mbox{\boldmath{$\alpha$}}_{k})=\sigma^2\mathbf{D}$, where $\mathbf{D}$ is some positive definite matrix.

Our main focus is the prediction of the genotype effects $\mbox{\boldmath{$\alpha$}}=(\mbox{\boldmath{$\alpha$}}_1, \dots, \mbox{\boldmath{$\alpha$}}_K)^\top$ and  the pairwise linear contrasts $\mbox{\boldmath{$\theta$}}^{k,k'}=\mbox{\boldmath{$\alpha$}}_k-\mbox{\boldmath{$\alpha$}}_{k'}$, $k, k'=1, \dots, K$, $k\neq k'$. For given total number of locations $J$ we search for the numbers of locations $J_1, \dots, J_P$ within the sub-regions, which are optimal for the prediction. Optimal designs for the estimation of fixed effects in LMM are well discussed in the literature (see e.g. \cite{fedo} or \cite{ent}). Less has been done for the prediction of random effects: the most general case - hierarchical random coefficient regression models - has been considered by \cite{pru1}. However, due to its more complicated covariance structure, model \eqref{mod} is not a particular case of those models. Therefore, the proposed approach cannot be used here. Also in the recently published work of \cite{pru7} a simpler covariance structure has been assumed.

We measure the performance of the prediction in terms of the mean squared error (MSE) matrix. The MSE matrices of the BLUPs for the genotype effects $\mbox{\boldmath{$\alpha$}}$ and for the pairwise linear contrasts $\mbox{\boldmath{$\theta$}}^{k,k'}$ are provided by the next lemma. 

\begin{lem}\label{lem}
a) The MSE matrix of the BLUP $\hat{\mbox{\boldmath{$\alpha$}}}$ of $\mbox{\boldmath{$\alpha$}}$ is given by
\begin{equation}\label{mse}
\mathrm{Cov}(\hat{\mbox{\boldmath{$\alpha$}}}-\mbox{\boldmath{$\alpha$}})=\sigma^2\left[\frac{1}{K}\mathds{1}_K\mathds{1}_K^\top\otimes \mathbf{D}+ (\mathbb{I}_K-\frac{1}{K}\mathds{1}_K\mathds{1}_K^\top)\otimes\left(\frac{1}{2v_2+1}\mathbf{F}^\top\mathbf{F}+\mathbf{D}^{-1}\right)^{-1}\right],
\end{equation}
where $\mathbf{F}=\textrm{block-diag}(\mathds{1}_{2J_1}, \dots, \mathds{1}_{2J_P})$, $\mathds{1}_{s}$ is the vector of length $s$ with all entries equal to $1$, $\mathbb{I}_s$ is the $s\times s$ identity matrix and $\otimes$ denotes the Kronecker product.

b) The MSE matrix of the BLUP $\hat{\mbox{\boldmath{$\theta$}}}^{k,k'}$ of $\mbox{\boldmath{$\theta$}}^{k,k'}$ is given by 
\begin{equation}\label{mse2}
\mathrm{Cov}(\hat{\mbox{\boldmath{$\theta$}}}^{k,k'}-\mbox{\boldmath{$\theta$}}^{k,k'})=2\sigma^2\left(\frac{1}{2v_2+1}\mathbf{F}^\top\mathbf{F}+\mathbf{D}^{-1}\right)^{-1}.
\end{equation}
\end{lem}
The proof of the lemma is deferred to the appendix.

Note that in Lemma~\ref{lem} the MSE matrix \eqref{mse2} is the same for all $k,k'$. Therefore, we can fix $k$ and $k'$ and use the simplified notation $\mbox{\boldmath{$\theta$}}$ instead of $\mbox{\boldmath{$\theta$}}^{k,k'}$.

Further, for a given total number of locations, $J$, we search for the numbers of locations within sub-regions, which minimize the MSE matrix \eqref{mse} or \eqref{mse2} of the prediction for the genotype effects or for the pairwise linear contrasts, respectively.

\section{Optimal Design}

For the present optimization problem, we define (exact) designs as follows:
\begin{equation}
\xi:=\left( \begin{array}{ccc} x_1 & ... & x_P \\ J_1 & ... & J_P
\end{array}\right),
\end{equation}
where $x_1, \dots, x_P$ denote the sub-regions.

For analytical purposes we also introduce approximate designs (see e.\,g. \cite{kie}):  
\begin{equation}
\xi:=\left( \begin{array}{ccc} x_1 & ... & x_P \\ w_1 & ... & w_P
\end{array}\right),
\end{equation}
where $w_i=J_i/J$ is the weight of locations within sub-region $i$. For these designs the requirement of integer values of $J_i$ is dropped and only the conditions $\sum_{i=1}^P{w_i}=1$ and $w_i\geq 0$ have to be satisfied.

We define the information matrix as
\begin{equation}
\mathbf{M}(\xi)=\mathrm{diag}(w_1, \dots, w_P)
\end{equation} 
and note that for exact designs the following condition is satisfied:
\begin{equation}
\mathbf{M}(\xi)=\frac{1}{2J}\mathbf{F}^\top\mathbf{F}.
\end{equation}

Then we extend the definitions of MSE matrices \eqref{mse} and \eqref{mse2} with respect to approximate designs and obtain
\begin{equation}\label{mse1}
\mathrm{MSE}_{\alpha}(\xi)=\sigma^2\left[\frac{1}{K}\mathds{1}_K\mathds{1}_K^\top\otimes \mathbf{D}+ (\mathbb{I}_K-\frac{1}{K}\mathds{1}_K\mathds{1}_K^\top)\otimes\left(\frac{2J}{2v_2+1}\mathbf{M}(\xi)+\mathbf{D}^{-1}\right)^{-1}\right]
\end{equation}
and
\begin{equation}\label{mse3}
\mathrm{MSE}_{\theta}(\xi)=2\sigma^2\left(\frac{2J}{2v_2+1}\mathbf{M}(\xi)+\mathbf{D}^{-1}\right)^{-1}.
\end{equation}

\subsection{A-optimal designs}

The A-criterion for prediction may be defined as the trace of the MSE matrix (see e.\,g. \cite{pru1}). For approximate designs this definition can be generalized using the extended MSE matrices \eqref{mse1} and \eqref{mse3}. Then we evaluate (neglecting the constant factor $\frac{\sigma^2}{J}(2v_2+1)$) for the pairwise linear contrasts $\mbox{\boldmath{$\theta$}}$ the criterion
\begin{equation}\label{acr}
\Phi_{A}(\xi)=\mathrm{tr}\left( \mathbf{M}(\xi)+\mathbf{\Delta}^{-1} \right)^{-1},
\end{equation}
where $\mathbf{\Delta}=\frac{2J}{2v_2+1}\mathbf{D}$. 

The A-criterion for the genotype effects $\mbox{\boldmath{$\alpha$}}$ differs from \eqref{acr} only by the constant term $\sigma^2\mathrm{tr}(\mathbf{D})$ and the multiplicator $\frac{\sigma^2}{2J}(K-1)(2v_2+1)$, both of which have no influence on the solution to the optimization. Therefore, optimal designs for the prediction of the genotype effects and the linear contrasts are the same. The next theorem provides the optimality condition for approximate designs.

\begin{thm}\label{oc}
An approximate design $\xi^*$ is A-optimal for the prediction of the genotype effects $\mbox{\boldmath{$\alpha$}}$ and the pairwise linear contrasts $\mbox{\boldmath{$\theta$}}$ iff
\begin{equation}\label{oca}
\mathrm{tr}\left( \mathbf{M}(\xi^*)\left(\mathbf{M}(\xi^*)+\mathbf{\Delta}^{-1}\right)^{-2} \right)\geq \mathbf{e}_i^\top\left( \mathbf{M}(\xi^*)+\mathbf{\Delta}^{-1} \right)^{-2}\mathbf{e}_i,\quad i=1, \dots, P,
\end{equation}
where $\mathbf{e}_i$ is the vector of length $P$ with the $i$-th entry equal to $1$ and all other entries equal to $0$. 

For all $i$ with $w_i>0$ equality holds in \eqref{oca}.
\end{thm}
\begin{proof}
The A-criterion \eqref{acr} may be recognized as a particular Bayesian A-criterion. The optimality condition follows from Theorem~1 in \cite{gla} for the linear transformation matrix $\mathbf{H}=\mathbb{I}_P$, the regression functions $\mathbf{f}(x_i)=\mathbf{e}_i$ and the design region $\mathcal{X}=\{x_1, \dots, x_p\}$.
\end{proof}
\begin{fol}
Let $\xi^*$ be an A-optimal design for the prediction of the genotype effects $\mbox{\boldmath{$\alpha$}}$ and the pairwise linear contrasts $\mbox{\boldmath{$\theta$}}$. Let $x_i$
 and $x_{i'}$ be support points of $\xi^*$ ($w_i>0$ and $w_{i'}>0$). Then the following equality holds:
\begin{equation}\label{sp1}
\mathbf{e}_i^\top\left( \mathbf{M}(\xi^*)+\mathbf{\Delta}^{-1} \right)^{-2}\mathbf{e}_i=\mathbf{e}_{i'}^\top\left( \mathbf{M}(\xi^*)+\mathbf{\Delta}^{-1} \right)^{-2}\mathbf{e}_{i'}.
\end{equation}
\end{fol}
Note that designs for which some weights $w_i$ (and consequently some numbers of locations $J_i$ per sub-region) have zero value are also acceptable in our research.
\vspace{0.3cm}

\textbf{Example~1: Compound symmetry model.}
We consider a (compound symmetry) model with the particular covariance structure of genotype effects $\mathbf{D}=a\,\mathds{1}_P\mathds{1}_P^\top+b\,\mathbb{I}_P$ with positive $b$ and real valued $a$ (for which $\mathbf{D}$ is positive definite). For this model some optimal designs can be obtained explicitly.
\begin{thm}\label{csa}
In the compound symmetry model the (balanced) design $\xi_P$ with $w_i=\frac{1}{P}$, $i=1, \dots, P$, is A-optimal for the prediction of the genotype effects and the pairwise linear contrasts.
\end{thm}
\begin{proof}
For the balanced design $\xi_P$ the information matrix is given by $\mathbf{M}(\xi_P)=\frac{1}{P}\,\mathbb{I}_P$. 
Then it can be easily verified that all diagonal entries of the matrix $\left(\mathbf{M}(\xi_P)+\mathbf{\Delta}^{-1}\right)^{-2}$ 
are the same, which leads to equalities in \eqref{oca} for all $i=1, \dots, P$.
\end{proof} 




\subsection{Optimal designs with respect to weighted A-criterion}

In this section we focus on the prediction of the pairwise contrasts $\mbox{\boldmath{$\theta$}}$. We define the weighted A-criterion as the weighted sum across all sub-regions of the variances of the differences between the predicted and the real contrasts:  
\begin{equation}
\Phi_{A_w, \theta}=\sum_{i=1}^P{\ell_i\mathrm{var}(\hat{\theta}_i-\theta_i)},
\end{equation}
where $\ell_1, \dots, \ell_P$ denote coefficients, which are related to the sub-regions. One possible choice is the size of the sub-regions. Alternatively, equal weight may be given to each sub-region, meaning that $\ell_i = \frac{1}{P}$ for all $i$. (In this case the weighted A-criterion coincides with the standard one).
Then we extend this definition with respect to approximate designs and obtain (neglecting the constant multiplicator $\frac{\sigma^2}{J}(2v_2+1)$) the following criterion:
\begin{equation}\label{wacr}
\Phi_{A_w, \theta}(\xi)=\mathrm{tr}\left(\mathbf{L} \left(\mathbf{M}(\xi)+\mathbf{\Delta}^{-1} \right)^{-1}\right),
\end{equation}
where $\mathbf{L}=\mathrm{diag}(\ell_1, \dots, \ell_P)$. The next theorem presents the optimality condition for approximate designs.

\begin{thm}
A design $\xi^*$ is optimal for the prediction of pairwise linear contrasts $\mbox{\boldmath{$\theta$}}$ with respect to the weighted A-criterion iff
\begin{eqnarray}\label{oca1}
&&\mathrm{tr}\left( \mathbf{M}(\xi^*)\left(\mathbf{M}(\xi^*)+\mathbf{\Delta}^{-1}\right)^{-1} \mathbf{L}\left(\mathbf{M}(\xi^*)+\mathbf{\Delta}^{-1}\right)^{-1}\right)\nonumber\\
&&\geq \mathbf{e}_i^\top\left(\mathbf{M}(\xi^*)+\mathbf{\Delta}^{-1}\right)^{-1}\mathbf{L}\left(\mathbf{M}(\xi^*)+\mathbf{\Delta}^{-1}\right)^{-1}\mathbf{e}_i,\quad i=1, \dots, P.
\end{eqnarray}
For all $i$ with $w_i>0$ equality holds in \eqref{oca1}.
\end{thm}
\begin{proof}
The weighted A-criterion \eqref{wacr} may be recognized as a particular Bayesian linear criterion. Optimality condition \eqref{oca1} follows from Theorem~1 in \cite{gla} for the linear transformation matrix $\mathbf{H}=\mathbf{L}$.  
\end{proof}
\begin{fol}
Let $\xi^*$ be an optimal design with respect to the weighted A-criterion for the prediction of the genotype effects $\mbox{\boldmath{$\alpha$}}$ and the pairwise linear contrasts $\mbox{\boldmath{$\theta$}}$. Let $x_i$
 and $x_{i'}$ be support points of $\xi^*$ ($w_i>0$ and $w_{i'}>0$). Then the following equality holds:
\begin{equation}\label{sp2}
\mathbf{e}_i^\top\left(\mathbf{M}(\xi^*)+\mathbf{\Delta}^{-1}\right)^{-1} \mathbf{L}\left(\mathbf{M}(\xi^*)+\mathbf{\Delta}^{-1}\right)^{-1}\mathbf{e}_i=\mathbf{e}_{i'}^\top\left(\mathbf{M}(\xi^*)+\mathbf{\Delta}^{-1}\right)^{-1} \mathbf{L}\left(\mathbf{M}(\xi^*)+\mathbf{\Delta}^{-1}\right)^{-1}\mathbf{e}_{i'}.
\end{equation}
\end{fol}
For the weighted A-criterion the optimal designs are not as easy to guess as for the standard A-criterion. Specifically, it is worth mentioning that the design $\xi_L$ with $w_i=\ell_i/\ell$, $\ell=\sum_{i=1}^P{\ell_i}$, which intuitively could be a solution of the optimization problem, is in general not optimal (see the real data example below).

\subsection{Enforcing the same efficiency in each sub-region}
For some studies it is required that the variances of the differences between the real and the predicted effects are the same for all sub-regions. For the model under investigation (model \eqref{mod}) this condition is given by
\begin{equation}\label{con}
\mathbf{e}_i^\top\left( \mathbf{M}(\xi^*)+\mathbf{\Delta}^{-1} \right)^{-1}\mathbf{e}_i=\mathbf{e}_{i'}^\top\left( \mathbf{M}(\xi^*)+\mathbf{\Delta}^{-1} \right)^{-1}\mathbf{e}_{i'},\quad i,i'=1, \dots, P.
\end{equation}
Under this condition the numbers of locations $J_1, \dots, J_P$ can be obtained as a solution of a system of $P-1$ equations (for example fix $i=1$ and $i'=2, \dots, P$) with $P-1$ unknown variables ($J_P=J-\sum_{i=1}^{P-1}J_i$) and no further optimization is needed.

%

%
\begin{bem}
For the compound symmetry model the balanced design is a solution of \eqref{con}.
\end{bem}
\begin{bem}
Not all A-optimal designs for the prediction of the genotype effects and the pairwise linear contrasts satisfy condition \eqref{con}.
\end{bem}
\begin{thm}
Let the covariance matrix of random effects be diagonal. Let $\tilde{\xi}$ be a design satisfying condition \eqref{con}. Then $\tilde{\xi}$ is A-optimal for the prediction of the genotype effects and the pairwise linear contrasts.
\end{thm}

\begin{proof}
If the matrix $\mathbf{D}$ is diagonal, the matrix $\left(\mathbf{M}(\xi)+\mathbf{\Delta}^{-1}\right)^{-1}$ is also diagonal. Then the entries of the matrix $\left(\mathbf{M}(\xi)+\mathbf{\Delta}^{-1}\right)^{-2}$ are equal to the squared entries of $\left(\mathbf{M}(\xi)+\mathbf{\Delta}^{-1}\right)^{-1}$. If the condition \eqref{con} is satisfied for a design $\tilde{\xi}$, the diagonal entries of $\left(\mathbf{M}(\xi)+\mathbf{\Delta}^{-1}\right)^{-2}$ have to be all the same, which leads to equalities in the optimality condition \eqref{oca} for all $i=1, \dots, P$.
\end{proof}

\begin{fol}
Let the covariance matrix of random effects be diagonal. Let $\tilde{\xi}$ be an A-optimal design for the prediction of the genotype effects and the pairwise linear contrasts with $w_i>0$ for all $i=1, \dots, P$. Then $\tilde{\xi}$ satisfies condition \eqref{con}.
\end{fol}

\section{Real Data Example}

We here consider variance components from a study on maize variety trials in India with five agroecological sub-regions. The dataset comprises four maturity groups of maize. Here, we consider only the extra-early maturity group (\cite{kle}, Tables 6 and 7). Based on the variance components reported in the paper, we derived variance components to be used in our design problems.

In Table~1 we summarize the variance components in the model under investigation (Model (1)) and how we determined these from the parameter estimates for the model considered by \cite{kle}. According to this table the adjusted covariance matrix $\mathbf{\Delta}$ form formula \eqref{acr} may be computed as
\begin{equation}
\mathbf{\Delta}=\frac{J}{493-\sigma^2/2}(\mathbf{V}+31\,\mathds{1}_5\mathds{1}_5^\top+18\,\mathbb{I}_5).
\end{equation}

\begin{table}
\caption{Variance components used in this example (column "Variance in Model (1)") and how they are derived from the variance parameter estimates in \cite{kle}}
\vspace{0.2cm}
\begin{tabular}{|p{2.6cm}|p{1.7cm}|p{3.65cm}|p{3.6cm}|p{4.0cm}|}
\hline
Effect & Model (1) & Model in Kleinknecht \textit{et al.} (2013) & Variance in Model (1) & Variance in Kleinknecht \textit{et al.} (2013)\\
\hline
Zone + mean & $\mu_i$  & $\mu+z_h+za_{hk}+a_k$ & fixed & 426+107+153 \\
\hline
Genotype$\times$zone & $\alpha_{ik}$ & $g_{i(h)}+gza_{ihk}+ga_{ik}$  & $\sigma^2\mathbf{D}$ for $\alpha_k$ & $\mathbf{V}+31\,\mathds{1}_5\mathds{1}_5^\top+18\,\mathbb{I}_5$\\
\hline
Location$\times$zone & $\lambda_{ij}+b_{ijl}$ & $l_{jh}+la_{jhk}$ & $\sigma^2_{\lambda}+\sigma^2_b=\sigma^2(v_1+v_3)$ & 1129+1000 \\
\hline 
Gen$\times$loc$\times$zone \newline + Obs errors & $\gamma_{ijk}+\varepsilon_{ijkl}$ & $gl_{ijh}+e_{hijk}$ & $\sigma^2_{\gamma}+\sigma^2=\sigma^2(v_2+1)$ & 160+333\\
\hline
\end{tabular}
\end{table}

We consider both the standard and the weighted A-criterion for first-order factor-analytic (FA) and compound symmetry (CS) variance-covariance structures, which were discussed in \cite{kle} (see also \cite{pie2} for FA models).

\subsection{Standard A-criterion}

For the first-order factor-analytic model we take the covariance matrix $\mathbf{V}$ from Table~6 in \cite{kle} (right part):
\begin{equation}\label{fa}
\mathbf{V}=\left(\begin{array}{ccccc}
	567 & 254 & 239 & 485 & 328 \\ 254 & 155 & 118 & 240 & 162 \\ 239 & 118 & 155 & 226 & 153 \\ 485 & 240 & 226 & 488 & 310 \\ 328 & 162 & 153 & 310 & 215
\end{array}\right).
\end{equation}

Table~2 summarizes results for optimal designs in the first-order factor-analytic model. As we can see in Table~2, optimal designs depend on both the total number of locations $J$ and the error variance $\sigma^2$. 
\begin{table}
\caption{Optimal numbers of locations per sub-region with respect to standard A-criterion in FA model for different values of the total number of locations $J$ and the error variance $\sigma^2$}
\vspace{0.2cm}
\centering
\begin{tabular}{|c|c|ccccc|ccccc|}
\hline
$J$ & $\sigma^2$ & \multicolumn{5}{c|}{Approximate design} & \multicolumn{5}{c|}{Exact design}\\
 & & $w_1$ & $w_2$ & $w_3$ & $w_4$ & $w_5$ & $J_1$ & $J_2$ & $J_3$ & $J_4$ & $J_5$\\
\hline
20 & 50 & 0.33 & 0.13 & 0.18 & 0.31 & 0.04 & 7 & 3 & 3 & 6 & 1\\
 & 200 & 0.31 & 0.15 & 0.19 & 0.29 & 0.06 & 6 & 3 & 4 & 6 & 1\\
 & 400 & 0.29 & 0.16 & 0.20 & 0.27 & 0.09 & 6 & 3 & 4 & 5 & 2\\
\hline
40 & 50 & 0.27 & 0.17 & 0.20 & 0.25 & 0.10 & 11 & 7 & 8 & 10 & 4\\
 & 200 & 0.26 & 0.18 & 0.20 & 0.24 & 0.12 & 10 & 7 & 8 & 10 & 5\\
 & 400 & 0.25 & 0.19 & 0.21 & 0.23 & 0.13 & 10 & 8 & 8 & 9 & 5\\
\hline
100 & 50 &  0.23 & 0.19 & 0.21 & 0.22 & 0.15 & 23 & 19 & 21 & 22 & 15\\
 & 200 & 0.23 & 0.19 & 0.21 & 0.21 & 0.16 & 23 & 19 & 21 & 21 & 16\\
 & 400 & 0.22 & 0.20 & 0.21 & 0.21 & 0.17 & 22 & 20 & 20 & 21 & 17\\
\hline
\end{tabular}
\end{table}
For the compound symmetry model the covariance matrix $\mathbf{V}$ is taken from Table 6 in \cite{kle} (left part, CS model):
\begin{equation}\label{cs}
\mathbf{V}=\left(\begin{array}{ccccc}
	308 & 270 & 270 & 270 & 270 \\ 270 & 308 & 270 & 270 & 270 \\ 270 & 270 & 308 & 270 & 270 \\ 270 & 270 & 270 & 308 & 270 \\ 270 & 270 & 270 & 270 & 308
\end{array}\right).
\end{equation}
For this model we obtain optimal designs $J_i=J/5$, $i=1, \dots, 5$, which is in accordance with Theorem~\ref{csa}.

\subsection{Weighted A-criterion}

For the weighted A-criterion we used the coefficients $\ell_1=813685$, $\ell_2=432716$, $\ell_3=477365$, $\ell_4=995298$, $\ell_5=1174818$, which correspond to the areas of the sub-regions, respectively, as determined from a digitized version of the map shown in \cite{kle}.

Table~3 summarizes the results for optimal designs in the factor-analytic model.
\begin{table}
\caption{Optimal numbers of locations per sub-region with respect to weighted A-criterion for FA model for different values of the total  number of locations $J$ and the error variance $\sigma^2$}
\vspace{0.2cm}
\centering
\begin{tabular}{|c|c|ccccc|ccccc|}
\hline
$J$ & $\sigma^2$ & \multicolumn{5}{c|}{Approximate design} & \multicolumn{5}{c|}{Exact design}\\
 & & $w_1$ & $w_2$ & $w_3$ & $w_4$ & $w_5$ & $J_1$ & $J_2$ & $J_3$ & $J_4$ & $J_5$\\
\hline
20 & 50 & 0.35 & 0.03 & 0.10 & 0.37 & 0.15 & 7 & 1 & 2 & 7 & 3\\
 & 200 & 0.33 & 0.05 & 0.11 & 0.35 & 0.16 & 7 & 1 & 2 & 7 & 3\\
 & 400 & 0.30 & 0.08 & 0.13 & 0.32 & 0.18 & 6 & 2 & 2 & 6 & 4\\
\hline
40 & 50 & 0.28 & 0.09 & 0.13 & 0.30 & 0.19 & 11 & 4 & 5 & 12 & 8\\
 & 200 & 0.27 & 0.10 & 0.14 & 0.29 & 0.20 & 11 & 4 & 6 & 11 & 8\\
 & 400 & 0.27 & 0.10 & 0.14 & 0.29 & 0.20 & 10 & 5 & 6 & 11 & 8\\
\hline
100 & 50 & 0.24 & 0.13 & 0.15 & 0.26 & 0.22 & 24 & 13 & 15 & 26 & 22\\
 & 200 & 0.24 & 0.13 & 0.15 & 0.25 & 0.22 & 24 & 13 & 15 & 25 & 23\\
 & 400 & 0.23 & 0.14 & 0.16 & 0.25 & 0.23 & 23 & 14 & 15 & 25 & 23\\
\hline
\end{tabular}
\end{table}
The optimal designs in the compound symmetry model are presented in Table~4.
\begin{table}
\caption{Optimal numbers of locations per sub-region with respect to weighted A-criterion for CS model for different values of the total number of locations $J$ and the error variance $\sigma^2$}
\vspace{0.2cm}
\centering
\begin{tabular}{|c|c|ccccc|ccccc|}
\hline
$J$ & $\sigma^2$ & \multicolumn{5}{c|}{Approximate design} & \multicolumn{5}{c|}{Exact design}\\
 & & $w_1$ & $w_2$ & $w_3$ & $w_4$ & $w_5$ & $J_1$ & $J_2$ & $J_3$ & $J_4$ & $J_5$\\
\hline
20 & 50 & 0.22 & 0.10 & 0.12 & 0.26 & 0.30 & 4 & 2 & 3 & 5 & 6\\
 & 200 & 0.21 & 0.11 & 0.12 & 0.26 & 0.30 & 4 & 2 & 3 & 5 & 6\\
 & 400 & 0.21& 0.11 & 0.13 & 0.26 & 0.29 & 4 & 2 & 3 & 5 & 6\\
\hline
40 & 50 & 0.21 & 0.12 & 0.13 & 0.25 & 0.29 & 9 & 5 & 5 & 10 & 11\\
 & 200 & 0.21 & 0.12 & 0.13 & 0.25 & 0.28 & 9 & 5 & 5 & 10 & 11\\
 & 400 & 0.21 & 0.13 & 0.14 & 0.25 & 0.28 & 9 & 5 & 5 & 10 & 11\\
\hline
100 & 50 & 0.21 & 0.13 & 0.14 & 0.24 & 0.27 & 21 & 13 & 15 & 24 & 27\\
 & 200 & 0.21 & 0.14 & 0.15 & 0.24 & 0.27 & 21 & 13 & 15 & 24 & 27\\
 & 400 & 0.21 & 0.14 & 0.15 & 0.24 & 0.26 & 21 & 14 & 15 & 24 & 26\\
\hline
\end{tabular}
\end{table}
Note that optimal designs in Tables~3 and 4 depend on the total number of locations $J$ and the error variance $\sigma^2$ and are in general not equal to the ratios $w_i=\ell_i/\ell$, $\ell=\sum_{i=1}^P{\ell_i}$. The results for the factor-analytic and compound symmetry models are different, illustrating that the optimal designs depend on the variance-covariance structure of genotypic effects within sub-regions. In case of compound symmetry optimal designs are less sensitive to $J$ and $\sigma^2$ than for the factor-analytic model.

All computations were performed using the procedures \textit{od.SOCP} and \textit{od.MISOCP} from the package \textit{OptimalDesign} in \textit{R}
for optimal approximate and exact designs, respectively, as proposed in \cite{harm}. Note that the exact designs obtained using \textit{od.MISOCP} are optimal in the class of exact designs for the model under investigation for the given data. 

\subsection{Enforcing the same efficiency in each sub-region}
When using the CS structure \eqref{cs} in $\mathbf{D}$, we obtained the trivial solution $w_i=0.2$, $i=1, \dots, 5$. When using the factor-analytic structure in \eqref{fa}, the solutions were as shown in Table~5. The exact designs were obtained by efficient rounding (see \cite{puk1}).
\begin{table}
\caption{
Optimal numbers of locations enforcing the same efficiency in each sub-region for FA model for different values of the total number of locations $J$ and the error variance $\sigma^2$}
\vspace{0.2cm}
\centering
\begin{tabular}{|c|c|ccccc|ccccc|}
\hline
$J$ & $\sigma^2$ & \multicolumn{5}{c|}{Approximate design} & \multicolumn{5}{c|}{Exact design}\\
 & & $w_1$ & $w_2$ & $w_3$ & $w_4$ & $w_5$ & $J_1$ & $J_2$ & $J_3$ & $J_4$ & $J_5$\\
\hline
20 & 50 & 0.342 &	0.148	& 0.205	& 0.302	& 0.003 & 6	& 3	& 4	& 6	& 1\\
 & 200 & 0.320	& 0.158	& 0.209	& 0.284	& 0.029 & 6	& 3	& 4	& 6	& 1\\
 & 400 & 0.291	& 0.172	& 0.211	& 0.260	& 0.065 & 6	& 3	& 4	& 5	& 2\\
\hline
40 & 50 & 0.274	& 0.179	& 0.212	& 0.247	& 0.088 & 11 & 7 & 8 & 10	& 4\\
 & 200 & 0.262	& 0.183	& 0.212	& 0.239	& 0.104 & 10 & 7 & 9 & 10	& 4\\
 & 400 & 0.247	& 0.189	& 0.211	& 0.228	& 0.125 & 10 & 8 & 8 & 9 & 5\\
\hline
100 & 50 & 0.231	& 0.194	& 0.209	& 0.217	& 0.150 & 23 & 19 & 21 & 22 & 15\\
 & 200 & 0.226 &	0.195	& 0.208	& 0.214	& 0.157 & 22 & 20 & 21 & 21 & 16\\
 & 400 & 0.219 & 0.197	& 0.206	& 0.210	& 0.167 & 22 & 20 & 20 & 21 & 17\\
\hline
\end{tabular}
\end{table}
We used the function \textit{nlphqn} in \textit{SAS/IML} to solve the nonlinear system of equations in \eqref{con}. 

\section{Discussion}
In crop research, the design of experiments is mainly considered in the context of a single environment and assuming that treatment effects are fixed (\cite{joh}; \cite{mea}). Recently, there has been an increased interest in design for experiments when treatments are modeled as correlated random effects using kinship or pedigree information (\cite{bue}, \cite{bue2}; \cite{cul}; \cite{but}). Also, the design of multi-environment trials has been considered in a few papers, most notably in the context of partially replicated (p-rep) trials (\cite{wil}), but also in broader contexts (\cite{gon}). To the best of our knowledge, however, the problem of allocation of location numbers in subdivided TPE has never been considered in any detail. The problem is reminiscent of optimal allocation in stage-wise sampling based on a nested random-effects model (\cite{sne}, p. 529) but the approach needed is more complex due to the linear mixed model involving several fixed and random effects and the optimization being targeted to the prediction of random effects. There is also some relation to small-area estimation in surveys where mixed models are used for estimation (\cite{jia}, \cite{tor}), but design in that context is not usually targeted at individual domains or small areas, and there is no notion of a larger number of treatments as in MET. 

Our main focus was the optimal allocation of locations for different sub-regions with respect to the estimation of genotype effects and pairwise linear contrasts for A- and particular linear (weighted A-) criteria. The proposed approach is based on the method of best linear unbiased prediction (BLUP). For our problem Bayesian optimal designs for a transformed covariance matrix of genotype effects turn out to be optimal. In the example we considered two kinds of models with respect to the covariance structure: first-order factor-analytic and compound symmetry. The resulting designs in both cases depend on the covariance structure, observational errors variance and the total number of locations in all sub-regions. The only exception is the standard A-criterion for compound symmetry: in this case balanced designs are optimal.

Our criterion integrates the efficiencies for BLUPs of interest across sub-regions. There are three variations to this approach. Two of them take a weighted or unweighted average across sub-regions, and optimization typically leads to allocations that imply unequal efficiency between sub-regions. The third approach imposes the additional restriction that efficiency be the same for each sub-region. We think this latter approach is particularly relevant when several administrative entities (federal states or countries) join forces to link up their trialling networks for cross-boundary analysis. For such efforts to be successful it is vital that the benefit, in terms of efficiency gain compared to independent analysis, can be split equally between the administrative entities involved.


\appendix

\section{Proof of Lemma~\ref{lem}}

To make use of the theoretical results that are available in the literature (see e.\,g. \cite{hen3}) for the prediction of random parameters we will represent the model \eqref{mod} as a particular case of the general LMM 
\begin{equation}\label{lmm}
\mathbf{Y}=\mathbf{X} \mbox{\boldmath{$\beta $}} + \mathbf{Z} \mbox{\boldmath{$\zeta$}} + \mbox{\boldmath{$\epsilon$}}
\end{equation}
with design matrices $\mathbf{X}$ and $\mathbf{Z}$ for the fixed effects and the random effects, respectively. In \eqref{lmm}, $\mbox{\boldmath{$\beta $}}$ denotes the fixed effects and $\mbox{\boldmath{$\zeta $}}$ are the random effects. The random effects and the observational errors $\mbox{\boldmath{$\epsilon $}}$ are assumed to have zero mean and to be all uncorrelated with positive definite covariance matrices $\mbox{Cov}\,(\mbox{\boldmath{$\zeta $}})=\mathbf{G}$ and $\mbox{Cov}\,(\mbox{\boldmath{$\epsilon $}})=\mathbf{R}$, respectively. Random effects and observational errors are assumed to be uncorrelated.

To present model \eqref{mod} in form \eqref{lmm} we use the following steps:
\begin{equation*}
\mathbf{Y}_{ijk}=\mathds{1}_2\,\mu_i+\mathds{1}_2\,\alpha_{ik}+\mathds{1}_2\,\lambda_{ij}+\mathds{1}_2\,\gamma_{ijk}+\mathbf{b}_{ij}+\mbox{\boldmath{$\varepsilon$}}_{ijk},\quad i=1 \dots P,\quad k=1, \dots, K,\quad j=t_{i-1}+1, \dots, t_i. 
\end{equation*}
\begin{equation*}
\mathbf{Y}_{ik}=\mathds{1}_{2J_i}\,\mu_i+\mathds{1}_{2J_i}\,\alpha_{ik}+(\mathbb{I}_{J_i}\otimes\mathds{1}_2)\,\mbox{\boldmath{$\lambda$}}_{i}+(\mathbb{I}_{J_i}\otimes\mathds{1}_2)\,\mbox{\boldmath{$\gamma$}}_{ik}+\mathbf{b}_{i}+\mbox{\boldmath{$\varepsilon$}}_{ik},\quad i=1 \dots P,\quad k=1, \dots, K,
\end{equation*}
where  $\mbox{\boldmath{$\lambda$}}_i=(\lambda_{i1}, \dots, \lambda_{iJ_i})^\top$ and $\mbox{\boldmath{$\gamma$}}_{ik}=(\gamma_{it_{i-1}+1k}, \dots, \gamma_{it_ik})^\top$.
\begin{equation*}
\mathbf{Y}_{k}=\mathbf{F}\mbox{\boldmath{$\mu$}}+\mathbf{F}\mbox{\boldmath{$\alpha$}}_{k}+\mathbf{H}\mbox{\boldmath{$\lambda$}}+\mathbf{H}\mbox{\boldmath{$\gamma$}}_k+\mathbf{b}+\mbox{\boldmath{$\varepsilon$}}_{k},\quad k=1, \dots, K,
\end{equation*}
where
$\mathbf{H}=(\mathbb{I}_J\otimes\mathds{1}_2)$, $\mbox{\boldmath{$\mu$}}=(\mu_1, \dots, \mu_P)^\top$, $\mbox{\boldmath{$\lambda$}}=(\mbox{\boldmath{$\lambda$}}_1^\top, \dots, \mbox{\boldmath{$\lambda$}}_P^\top)^\top$ and $\mbox{\boldmath{$\gamma$}}_k=(\mbox{\boldmath{$\gamma$}}_{1k}^\top, \dots, \mbox{\boldmath{$\gamma$}}_{Pk}^\top)^\top$.
\begin{equation*}
\mathbf{Y}=(\mathds{1}_K\otimes\mathbf{F})\mbox{\boldmath{$\mu$}}+(\mathbb{I}_K\otimes\mathbf{F})\mbox{\boldmath{$\alpha$}}+(\mathds{1}_K\otimes\mathbf{H})\mbox{\boldmath{$\lambda$}}+(\mathbb{I}_K\otimes\mathbf{H})\mbox{\boldmath{$\gamma$}}+(\mathds{1}_K\otimes\mathbb{I}_{2J})\mathbf{b}+\mbox{\boldmath{$\varepsilon$}},
\end{equation*}
where $\mbox{\boldmath{$\gamma$}}=(\mbox{\boldmath{$\gamma$}}_1, \dots, \mbox{\boldmath{$\gamma$}}_K)^\top$.

The latter equation may be alternatively written as
\begin{equation}\label{mod1}
\mathbf{Y}=(\mathds{1}_K\otimes\mathbf{F})\mbox{\boldmath{$\mu$}}+(\mathbb{I}_K\otimes\mathbf{F})\mbox{\boldmath{$\alpha$}}+\tilde{\mbox{\boldmath{$\varepsilon$}}},
\end{equation}
where $\tilde{\mbox{\boldmath{$\varepsilon$}}}:=(\mathds{1}_K\otimes\mathbf{H})\mbox{\boldmath{$\lambda$}}+(\mathbb{I}_K\otimes\mathbf{H})\mbox{\boldmath{$\gamma$}}+(\mathds{1}_K\otimes\mathbb{I}_{2J})\mathbf{b}+\mbox{\boldmath{$\varepsilon$}}$. Model \eqref{mod1} is of form \eqref{lmm} with $\mathbf{X}=(\mathds{1}_K\otimes\mathbf{F})$, $\mathbf{Z}=(\mathbb{I}_K\otimes\mathbf{F})$, $\mathbf{G}=\mathrm{Cov}(\mbox{\boldmath{$\alpha$}})=\sigma^2\mathbb{I}_K\otimes\mathbf{D}$ and
\begin{equation*}
\mathbf{R}=\mathrm{Cov}(\tilde{\mbox{\boldmath{$\varepsilon$}}})=\sigma^2((v_1\mathds{1}_K\mathds{1}_K^\top+v_2\mathbb{I}_K)\otimes\mathbb{I}_J\otimes\mathds{1}_2\mathds{1}_2^\top+(\mathds{1}_K\mathds{1}_K^\top\otimes\mathbb{I}_{2J})v_3+\mathbb{I}_{2JK}).
\end{equation*} 

According to \cite{hen3} the MSE matrix of the BLUP of the random effects $\mbox{\boldmath{$\zeta$}}$ (which corresponds to $\mbox{\boldmath{$\alpha$}}$ in our model \eqref{mod1}) is given by
\begin{equation}\label{c22}
\mathrm{Cov}(\hat{\mbox{\boldmath{$\zeta$}}}-\mbox{\boldmath{$\zeta$}})=\left(\mathbf{Z}^\top \mathbf{R}^{-1}\mathbf{Z} +\mathbf{G}^{-1}-\mathbf{Z}^\top \mathbf{R}^{-1}\mathbf{X} (\mathbf{X}^\top \mathbf{R}^{-1}\mathbf{X})^{-}\mathbf{X}^\top \mathbf{R}^{-1}\mathbf{Z}\right)^{-1},
\end{equation}
where $\mathbf{A}^{-}$ denotes a generalized inverse of $\mathbf{A}$. 
With this formula we obtain MSE matrix \eqref{mse}. 
Then using the relation $\mbox{\boldmath{$\theta$}}^{k,k'}=((\mathbf{e}_k-\mathbf{e}_{k'})^\top\otimes\mathbb{I}_P)\,\mbox{\boldmath{$\alpha$}}$ between the genotype effects and the pairwise contrasts we receive formula \eqref{mse2}.

\section*{Acknowledgment}

This research was partially supported by grant SCHW 531/16 of the German Research Foundation (DFG).
The authors are grateful to Waqas Malik (University of Hohenheim) for determining the areas of the five breeding zones for maize in India based on a digitized map.

\bibliographystyle{natbib}
\bibliography{prus6}

\end{document}